\documentclass[aps,pra,twocolumn,notitlepage]{revtex4-1}

\usepackage{amsmath}
\usepackage{braket}
\usepackage{amsfonts}
\usepackage{amssymb}
\usepackage{amsthm}
\usepackage{optidef}

\newcommand{\dyad}[1]{\ket{#1}\!\bra{#1}}
\newcommand{\Tr}{\operatorname{Tr}}
\newcommand{\id}{\operatorname{id}}

\newtheorem{theorem}{Theorem}
\newtheorem{lemma}{Lemma}

\newtheorem{observation}{Observation}

\newcommand{\cA}{\mathcal{A}}
\newcommand{\cH}{\mathcal{H}}
\newcommand{\cS}{\mathcal{S}}

\begin{document}
	\title{A simple class of bound entangled states based on the properties of the antisymmetric subspace}
	
	\author{Enrico Sindici}
	\affiliation{SUPA and Department of Physics, University of Strathclyde, Glasgow, G4 0NG, UK}
	\author{Marco Piani}
	\affiliation{SUPA and Department of Physics, University of Strathclyde, Glasgow, G4 0NG, UK}

	\begin{abstract}
		We provide a simple construction of bipartite entangled states that are positive under partial transposition, and hence undistillable. The construction makes use of the properties of the projectors onto the symmetric and antisymmetric subspaces of the Hilbert space of two identical systems. The resulting states can be considered as generalizations of the celebrated Werner states.
	\end{abstract}

	\maketitle
	
	\section{Introduction}
	
	Entanglement~\cite{Horodecki:2009aa} is at the core of quantum information processing~\cite{Nielsen:2011aa}. By considering the physically motivated framework of distant laboratories, where only transformations implemented by local operations and classical communication (LOCC) are allowed, entanglement is elevated to the status of resource. In such a framework, there is great interest in understanding possibilities and limitations in the manipulation of entanglement, in particular with respect to the distillation of noisy entanglement into pure-state entanglement~\cite{Bennett:1996aa}. We know that noisy entangled states that are positive under partial transposition (PPT) cannot be distilled~\cite{Horodecki:1998aa}, and are hence called bound entangled~\footnote{Whether it is only PPT entangled states that cannot be distilled is one of the last major open questions in entanglement theory.}. To focus on such a noisy kind of entanglement is useful and interesting for several reasons. One such reason is that noisy entangled states provide a testbed for entanglement detection methods~\cite{GUHNE20091}, and for the study of the relation between phenomena like entanglement, steering~\cite{Moroder:2014aa}, and non-locality~\cite{Vertesi:2014aa}. From a mathematical standpoint, the study of such states is linked to the study of positive but not completely positive maps~\cite{HORODECKI19961}. Finally, noisy entangled states are also connected to superactivation effects in quantum information~\cite{Smith1812}.
	
	While there are several examples of PPT entangled states in literature (see, e.g., \cite{HORODECKI1997333,Bennett:1999aa,Benatti:2004aa,Piani:2007aa,Sentis:2016aa}), their structure is often relatively complicated, and not amenable to a simple parametrization in terms of a noisy parameter or dimensionality. From a theoretical perspective, one consequence of this is that, when discussing how noise affects tasks and tests that involve entanglement, the analysis of the role of noise is often less comprehensive than it could be. This is because it is customary to focus on exemplary classes of noisy entangled states with a simple structure, like Werner states~\cite{Werner:1989aa} (see Section~\ref{sec:Werner}) or isotropic states~\cite{Horodecki:1999aa}, which do not exhibit PPT entanglement in any range of the parameter involved. From an experimental point of view, it would be convenient to have examples of PPT entangled states with a simple structure, because they could conceivably be implemented more easily in the laboratory (see, e.g., \cite{Amselem:2009aa,Lavoie:2010aa,Barreiro:2010aa,DiGuglielmo:2011aa}), for example to test effects like superactivation~\cite{Smith1812}.
	
	Here we present some tools for the numerical and analytical construction of simple examples of PPT bound entangled states. Such examples are based on the properties of the projections onto the symmetric and antisymmetric spaces of two qudits.

	\section{Entanglement, partial transposition, and bound entangled states}
	
	We recall some basic notions of entanglement theory~\cite{Horodecki:2009aa}.
	
	Consider two finite-dimensional systems $A$ and $B$, with Hilbert spaces $\cH_A \simeq \mathbb{C}^{d_A}$ and $\cH_B \simeq \mathbb{C}^{d_B}$, respectively. The joint Hilbert space is $\cH_{AB} = \cH_A\otimes\cH_B$. A factorized vector state $\ket{\alpha}\ket{\beta}\equiv\ket{\alpha}_A\otimes\ket{\beta}_B$ is called unentangled. Any vector state state $\ket{\psi}\equiv\ket{\psi}_{AB}\in\cH_{AB}$ that is not unentangled is entangled. Any vector state of the joint system can always can be written in the Schmidt decomposition form
	\begin{equation}
	\label{eq:schmidt}
	\ket{\psi} = \sum_i \sqrt{p_i}\ket{a_i}\ket{b_i},
	\end{equation}
	for an appropriate choice of orthonormal bases $\{\ket{a_i}\}$ for $A$ and $\{\ket{b_i}\}$ for $B$, with the $p_i$'s forming a probability distribution. The number of non-zero terms in such a probability distribution, that is, the number of non-zero factorized terms that enter in the Schmidt decomposition, is called the Schmidt rank of $\ket{\psi}$. A general mixed state $\rho$ of $AB$ is a positive-semidefinite unit-trace operator on $\cH_{AB}$, and it can be expressed as convex combination of projectors onto pure states $\dyad{\psi}$:
	\[
	\rho = \sum_k q_k \dyad{\psi_k}.
	\]
	We say that $\rho$ has Schmidt number $m$ if it can be expressed as convex combination of pure states such that each $\ket{\psi_k}$ has at most Schmidt rank $m$, and if any convex combination corresponding to $\rho$ necessarily contains at least one state $\ket{\psi_k}$ with Schmidt rank greater or equal to $m$ (with non-vanishing probability)~\cite{Terhal:2000aa}.
	
	A mixed state is separable or unentangled if it has Schmidt number one, that is, if it can be expressed as
	\begin{equation}
	\label{eq:separable}
	\rho = \sum_k q_k \dyad{\alpha_k}\otimes\dyad{\beta_k}.
	\end{equation}
	Notice that in such a separable expression, the states $\ket{\alpha_k}$ ($\ket{\beta_k}$) do not necessarily correspond to an orthonormal basis for $A$ (for $B$). A mixed state is entangled if it has Schmidt number strictly larger than one, equivalently, if it is not of the form \eqref{eq:separable}. In general, it is hard to determine whether a mixed state is separable or entangled~\cite{GUHNE20091,Horodecki:2009aa}. A simple but powerful test to detect entanglement is given by partial transposition~\cite{Peres:1996aa,HORODECKI19961}: if the state $\rho$ is separable, then the partially transposed state $\rho^{\Gamma_A}=(T_A\otimes \id_B)[\rho]$, where $T$ indicates the transposition operation, is still a positive semidefinite operator; thus if $\rho^{\Gamma_A}$ is not positive semidefinite, then $\rho$ must be entangled. The basis of $A$ in which partial transposition is taken is irrelevant for the sake of the power of the test, because one easily verifies that the spectrum of the partially transposed state does not depend on such a choice. Similarly, one could equivalently apply partial transposition on $B$, because $(\cdot)^{\Gamma_B} = ((\cdot)^{\Gamma_A})^{T_{AB}}$, where $T_{AB}$ is a global transposition that preserves positivity. Hence, in the following, we will indicate the partially transposed state simply by $\rho^\Gamma$, unless further specification is required.
	
	Many protocols in quantum information processing make use of pure-state entanglement, or, even more specifically, of maximally entangled states, where the probability distribution in Eq. \eqref{eq:schmidt} is flat. Since entangled states that are generated between distant locations are rarely of this form, an important process in entanglement manipulation is that of entanglement distillation, where many copies of a mixed entangled state $\rho$ are transformed into many (approximate) copies of a maximally entangled state at some rate. When the rate of conversion is non-zero, we say that the state $\rho$ is distillable, while entangled states such that the rate vanishes are called undistillable. One proves that PPT entangled states are undistillable~\cite{Horodecki:1998aa}.

	\section{Symmetric and antisymmetric subspace}
	
	Let $A$ and $B$ be two $d$-dimensional systems, with total Hilbert space $\cH_{AB}\simeq \mathbb{C}^d\otimes \mathbb{C}^d$. Such a composite Hilbert space can be written as the direct sum of the symmetric and antisymmetric subspaces: $\cH_{AB} = \cH_{\cS}\oplus\cH_{\cA}$, with the symmetric subspace $\cH_{\cS}=\mathbb{C}^d\vee\mathbb{C}^d$ and the antisymmetric subspace $\cH_{\cA}=\mathbb{C}^d\wedge\mathbb{C}^d$. The symbol $\oplus$ denotes direct sum, while $\vee$ and $\wedge$ denote  the symmetric and antisymmetric tensor product, respectively~\cite{bhatia:1997aa}. The symmetric and antisymmetric subspaces have dimensions $d_\cS := d(d+1)/2$ and $d_\cA := d(d-1)/2$, respectively.
	Consider the swap operator $V=V_{AB}$ which may be defined implicitly by its swapping action on every factorized state $\ket{\alpha}\ket{\beta}$: $V\ket{\alpha}\ket{\beta}=\ket{\beta}\ket{\alpha}$. It is worth noticing that, given the maximally entangled state
	\[
	\ket{\psi^+}=\frac{1}{\sqrt{d}}\sum_{i=1}^d \ket{i}\ket{i},
	\]
	for $\{\ket{i}\}$ a chosen local computational orthonormal basis, one has
	\[
	\dyad{\psi^+}^\Gamma = d V,
	\] 
	where partial transposition is taken in the computational basis.
	
	The projector onto the symmetric space is given by
	\begin{equation}
	\label{eq:Ps}
	P_\cS = \frac{\openone+V}{2},
	\end{equation}
	and the projector onto the antisymmetric space is given by
	\begin{equation}
	\label{eq:Pa}
	P_\cA = \frac{\openone-V}{2}.
	\end{equation}
	Such projectors are orthogonal and sum up to the identity operator. From their expressions \eqref{eq:Ps} and \eqref{eq:Pa}, it is immediate to derive the following relations for normalized single-system state vectors $\ket{\alpha}$ and $\ket{\beta}$:
	\begin{align}
	P_\cS \ket{\alpha}\ket{\beta} &= \frac{1}{2}\left(\ket{\alpha}\ket{\beta} + \ket{\beta}\ket{\alpha}\right) \\
	P_\cA \ket{\alpha}\ket{\beta} &= \frac{1}{2}\left(\ket{\alpha}\ket{\beta} - \ket{\beta}\ket{\alpha}\right) \label{eq:Paalphabeta}\\
	\bra{\alpha}\bra{\beta}P_\cS\ket{\alpha}\ket{\beta} &= \frac{1+|\braket{\alpha|\beta}|^2}{2}\\
	\bra{\alpha}\bra{\beta}P_\cA\ket{\alpha}\ket{\beta} &= \frac{1-|\braket{\alpha|\beta}|^2}{2} \label{eq:propPaalphabeta}.
	\end{align}
	
	\subsection{Werner states}
	\label{sec:Werner}
	
	Werner introduced a class of states that are invariant under $U\otimes U$ transformations~\cite{Werner:1989aa}. They correspond to convex combinations of the normalized projectors onto the symmetric and antisymmetric subspaces:
	\begin{equation}
	\rho_{\textrm{W}}(p) = p \frac{P_\cA}{d_\cA} + (1-p) \frac{P_\cS}{d_\cS}.
	\end{equation}
	It is well known that the Werner states are separable only when the probability $p$ is such that they are PPT, that is, for $0\leq p \leq 1/2$. This means that Werner states cannot be PPT bound entangled. On the other hand, it is an open question whether there is a range of values for $p$ such that entangled Werner states are undistillable even if non-positive under partial transposition (NPT)~\cite{DiVincenzo:2000aa,Shor:2001aa,Pankowski:2010aa}. Werner states, in particular the state $\rho_{\textrm{W}}(p=1)=P_\cA/d_\cA$, are characterized by very interesting properties, like their high degree of symmetry and the possibility of mapping non-trivially any state into a Werner state via LOCC, their not exhibiting Bell nonlocality~\cite{Werner:1989aa}, their high degree of shareability despite their degree of entanglement~\cite{Lancien:2016aa}, their implementing quantum data hiding~\cite{DiVincenzo:2002aa,DiVincenzo2003}.

	\section{Antisymmetric image of separable states}

	We provide two simple observations about the properties of projections of product states onto the antisymmetric subspace that follow directly by inspection from relations \eqref{eq:Paalphabeta} and \eqref{eq:propPaalphabeta}.
	
	\begin{observation}
		Let $\ket{\alpha}\ket{\beta}$ be a normalized product state. Then, either $\|P_\cA \ket{\alpha}\ket{\beta}\|=0$ in the case $|\braket{\alpha|\beta}|=1$, or $\ket{\psi}_{AB} = P_\cA \ket{\alpha}\ket{\beta}/\|P_\cA \ket{\alpha}\ket{\beta}\|$ is a normalized state with Schmidt rank equal to two.
	\end{observation}

	\begin{observation}
		\label{obs:Paalphabeta}
		Let $\ket{\psi}_{AB} = P_\cA \ket{\alpha}\ket{\beta}/\|P_\cA \ket{\alpha}\ket{\beta}\|$, for $\ket{\alpha}\ket{\beta}$ a normalized product state with $|\braket{\alpha|\beta}|<1$. Then there is another normalized state $\ket{\alpha'}\ket{\beta'}$ such that $\ket{\psi}_{AB} = \sqrt{2} P_\cA \ket{\alpha'}\ket{\beta'}$, with $\|P_\cA \ket{\alpha'}\ket{\beta'}\|=1/\sqrt{2}$. In particular this is possible with the choice $\ket{\alpha'}=\ket{\alpha}$ and  $\ket{\beta'}=\ket{\alpha^\perp}=(\ket{\beta}-\braket{\alpha|\beta}\ket{\alpha})/\|\ket{\beta}-\braket{\alpha|\beta}\ket{\alpha}\|$.
	\end{observation}

	Observation \ref{obs:Paalphabeta} leads to the following lemma.

	\begin{lemma}
		\label{lem:Paseparable}
		Let $\rho^{\textrm{sep}}$ be a separable state such that $\Tr(P_\cA \rho^{\textrm{sep}})>0$. Then there is a separable state $\rho'^{\textrm{sep}}$ such that $\Tr(P_\cA \rho'^{\textrm{sep}})=1/2$ and $\frac{P_\cA \rho^{\textrm{sep}} P_\cA}{\Tr(P_\cA \rho^{\textrm{sep}})}= \frac{P_\cA \rho'^{\textrm{sep}} P_\cA}{\Tr(P_\cA \rho'^{\textrm{sep}})}$.
	\end{lemma}
	\begin{proof}
		Let $\rho^{\textrm{sep}}=\sum_ip_i\dyad{\alpha_i}\otimes\dyad{\beta_i}$. To any term in the sum such that $p_i>0$ and $|\braket{\alpha_i|\beta_i}|<1$, associate a probability $p_i' = p_i \frac{\bra{\alpha_i}\bra{\beta_i}P_\cA \ket{\alpha_i}\ket{\beta_i}}{\Tr(\rho^{\textrm{sep}}P_\cA)}$ and local states $\ket{\alpha'_i}=\ket{\alpha_i}$, $\ket{\beta'_i}=(\ket{\beta}-\braket{\alpha|\beta}\ket{\alpha})/\|\ket{\beta}-\braket{\alpha|\beta}\ket{\alpha}\|$. Then the separable state $\rho'^{\textrm{sep}} = \sum_i p'_i \dyad{\alpha_i'}\otimes\dyad{\beta_i'}$ verifies the stated conditions, as it can be checked by the application of Observation~\ref{obs:Paalphabeta}.
	\end{proof}

	\section{A semidefinite-program approach to generate PPT entangled states}

	Let $\rho_\cA$ be a bipartite antisymmetric state, that is, fully supported in the antisymmetric subspace: $\rho_\cA=P_\cA \rho_\cA P_\cA$.
	
	We are interested in finding the largest probability of obtaining such a state from a PPT state by projecting onto the antisymmetric subspace, that is the following quantity, defined as the solution to an SDP:
	
	\begin{equation}
	\label{eq:SDP}
	\begin{aligned}
	p^\textrm{PPT}(\rho_\cA)
	 :=\quad&     	\underset{\sigma}{\text{max}}	& & \Tr(P_\cA\sigma) \\
	 & 	\text{s.t.} 			& &	P_\cA\sigma P_\cA = \Tr(P_\cA\sigma) \rho_\cA\\
	 & 								& &	\sigma \geq 0 \\
	 &								& &  \Tr(\sigma)=1 \\
	 &								& &  \sigma^\Gamma \geq 0.
	\end{aligned}
	\end{equation}
	
	We prove the following.
	
	\begin{theorem}
		\label{thm:bounds}
		It holds $2/(d(d+1)+2)\leq p^\textrm{PPT}(\rho_\cA) \leq 1/2$ for all antisymmetric states $\rho_\cA$. 
	\end{theorem}
	\begin{proof}
		Let us start from the lower bound. For the given $\rho_\cA$, let us consider the family of states $\sigma(p) = p \rho_\cA + (1-p)P_\cS/d_\cS$. By construction, $\sigma(p)$ is a valid quantum state, and it holds that $P_\cA\sigma(p) P_\cA = p \rho_\cA$, with $\Tr(P_\cA\sigma(p)) =p$. We now want to find a $\bar{p}$ such that $\sigma(p)^\Gamma \geq 0$ for all $p\leq \bar{p}$. One has
		\[
		\sigma(p)^\Gamma = p \rho_\cA^\Gamma + (1-p) \frac{\openone+\dyad{\psi^+}}{2d_\cS}.
		\]
		Thus, one finds
		\[
		\begin{aligned}
		&\quad\min_{\ket{\phi}} \braket{\phi|\sigma(p)^\Gamma|\phi} \\
		&=\min_{\ket{\phi}}\{p \braket{\phi|\rho_\cA^\Gamma|\phi} + \frac{1-p}{2d_\cS}(1+d|\braket{\phi|\psi^+}|^2) \} \\
		&\geq \min_{\ket{\phi}}\{p \braket{\phi|\rho_\cA^\Gamma|\phi} +
			 \frac{1-p}{2d_\cS} \} \\
		&\geq \frac{1}{2}\left(-p +\frac{1-p}{d_\cS}\right),
		\end{aligned}
		\]
		where we have used that
		\[
		\begin{aligned}
		&\quad\min_{\ket{\phi}}\braket{\phi|\rho_\cA^\Gamma|\phi} \\
		&\geq\min_{\ket{\phi},\rho_\cA}\braket{\phi|\rho_\cA^\Gamma|\phi}\\
		&= \min_{\ket{\phi},\rho_\cA}\Tr\left(\dyad{\phi}^\Gamma\rho_\cA\right)\\
		&= -\frac{1}{2},
		\end{aligned}
		\]
		since the smallest eigenvalue of the partial transposition of a pure state is at most $-1/2$, and its corresponding eigenstate is antisymmetric.
		Imposing $\frac{1}{2}(-p +(1-p)/d_\cS)\geq 0$ one finds $p\leq 1/(d_\cS+1)=2/(d(d+1)+2)=:\bar{p}$.
		
		The upper bound can be found by considering that, for an arbitrary PPT state $\sigma$, that is, such that $\sigma^\Gamma\geq0$, one has
		\[
		\Tr(P_\cA \sigma) = \frac{1}{2}(1-\Tr(V\sigma)) = \frac{1}{2}(1-d\braket{\psi^+|\sigma^\Gamma|\psi^+}) \leq \frac{1}{2}.
		\]
	\end{proof}

	Now suppose that, for a given antisymmetric state $\rho_\cA$, we find $p^\textrm{PPT}(\rho_\cA) < 1/2$, and that the optimal PPT state achieving the value is $\sigma^*$. We argue that $\sigma^*$ is a PPT entangled state. Indeed, suppose that it was separable; then, Lemma \ref{lem:Paseparable} ensures that this would imply the existence of some other separable state, which is a fortiori PPT, that would also be projected onto $\rho_\cA$ with probability $1/2$. This is a contradiction, since we have assumed $p^\textrm{PPT}(\rho_\cA) < 1/2$.
	
	Thus, one can generate PPT entangled states through the following procedure:
	\begin{enumerate}
		\item take an arbitrary antisymmetric state $\rho_\cA$;
		\item compute $p^\textrm{PPT}(\rho_\cA)$ via the SDP \eqref{eq:SDP};
		\item if $p^\textrm{PPT}(\rho_\cA)<1/2$, then the optimal PPT state $\sigma^*$ that is such that $P_\cA\sigma^* P_\cA = p^\textrm{PPT}(\rho_\cA) \rho_\cA$ is a PPT entangled state.
	\end{enumerate}
	Notice that antisymmetric states $\rho_\cA$ can be generated at random, for example by generating a random bipartite state $\rho$, and considering $\rho_\cA = P_\cA\rho P_\cA / \Tr(P_\cA\rho)$. 	

	\section{Structure of PPT states that generate an antisymmetric state}

	In the previous section we obtained a lower bound to $p^\textrm{PPT}(\rho_\cA)$ through the use of the class of feasible solutions for the SDP \eqref{eq:SDP} given by $\sigma(p) = p \rho_\cA + (1-p)P_\cS/d_\cS=p \rho_\cA \oplus (1-p)P_\cS/d_\cS$, which we proved to be PPT states for $p$ small enough. We argue here that, among the $PPT$ states $\sigma^*$ that are optimal for the sake of the probability  $p^\textrm{PPT}(\rho_\cA)$ defined in \eqref{eq:SDP}, there are always states with the structure $\sigma^* = p^\textrm{PPT}(\rho_\cA) \rho_\cA \oplus (1-p^\textrm{PPT}(\rho_\cA)) \rho_\cS$, where $\rho_\cS$ is a state with support on $\cH_\cS$. Indeed, let $\sigma^*$ be a PPT state that is optimal for the sake of $p^\textrm{PPT}(\rho_\cA)$. One can then consider $\sigma'^* = (\sigma^* + V\sigma^*V)/2$, which by construction has the structure $\sigma'^*=P_\cA \sigma^* P_A \oplus P_\cS \sigma^* P_S $. Notice that $(V\tau V)^{\Gamma_A} = V\tau^{\Gamma_B}V$, so that $V\tau V$ is PPT if and only if $\tau$ is PPT. Hence, $\sigma'^*$ is PPT, because it is the convex combination of two PPT states, and clearly such that $P_\cA \sigma'^* P_A = P_\cA \sigma^* P_A=p^\textrm{PPT}(\rho_\cA) \rho_\cA$.
	
	\section{Analytical examples of PPT entangled states}
	\label{sec:analytic}
	
	We want to provide analytical examples of PPT entangled states that can be identified as such based on reasoning along the lines of the previous sections. The idea is to look at states of the form $\sigma = p \rho_\cA \oplus (1-p) \rho_\cS$ for some simple choice of parameter $p$ and of states $\rho_\cA$ and $\rho_\cS$ that make the state $\sigma$ certifiably PPT entangled. We are going to choose $\rho_\cS = P_\cS / d_\cS$. From the proof of Theorem \ref{thm:bounds}, we know already that, as long as $p\leq 1/(d_\cS+1)$, $\sigma$ is going to be PPT. We only need to find a simple condition on $\rho_\cA$ that ensures that $\sigma$ is entangled. We can find such a condition invoking Lemma \ref{lem:Paseparable}, which implies that any antisymmetric pure state that originates from the projection onto the antisymmetric subspace of a pure factorized state has at most Schmidt rank equal to two. In general, this means that any separable state will be mapped onto antisymmetric mixed states of Schmidt number at most equal to two~\footnote{Notice that all antisymmetric mixed states have at least Schmidt number two, because the antisymmetric subspace does not contain product states, as it can be verified by using (\ref{eq:propPaalphabeta}).}. We conclude that, as soon as $\rho_\cA$ has Schmidt number strictly greater than two, and for any $p>0$, the state $p \rho_\cA \oplus (1-p) \rho_\cS$ is entangled. Notice that this does not contradict the fact that there are Werner states that are separable. In the case of Werner states, one has $\rho_\cA = P_\cA/d_\cA$, and the latter antisymmetric state, proportional to the projector onto the antisymmetric space, has Schmidt number equal to two, and can be obtained with probability $1/2$ from a separable state.
	
	The simplest way to make sure that $\rho_\cA$ has Schmidt number strictly larger than two is to choose $\rho_\cA = \dyad{\psi_\cA}$, for $\ket{\psi_\cA}$ an antisymmetric vector state with Schmidt rank strictly larger than two.
	
	We remark that generic random antisymmetric vector states in dimension $d=2m$ have Schmidt rank $2m$, and can in principle be generated (up to normalization) starting from a generic vector states without a definite symmetry, and projecting onto the antisymmetric space.
	
	An analytical, non-random constructions can be easily put forward. For example, in even dimensions $d=2m$, one can consider the antisymmetric vector states
    \begin{equation}
    \label{eq:psiAmultilevel}
    \ket{\psi_\cA} = \sum_{i=1}^{m} c_i \ket{\psi^-_{2i-1,2i}},\quad \sum_{i=1}^m |c_i|^2 =1,
    \end{equation}
    where
    \[
    \ket{\psi^-_{k,l}} = \frac{1}{\sqrt{2}}(\ket{k}\ket{l}-\ket{l}\ket{k}),
    \]
    for $k,l=1,2,\ldots,d$ and $k<l$.
    The vector state $\ket{\psi_\cA}$ in Eq.~\eqref{eq:psiAmultilevel} has Schmidt rank equal to twice the number of non-zero amplitudes $c_i$. One may consider of particular interest the case where $\ket{\psi_\cA}$ in Eq.~\eqref{eq:psiAmultilevel} is maximally entangled between that two $2m$-dimensional systems, that is, where the coefficients satisfy $|c_i|=1/\sqrt{m}$ so that all its Schmidt coefficients are equal to $1/\sqrt{d}=1/\sqrt{2m}$.

	\section{Conclusions}
	
	We constructed simple analytical examples of entangled states that remain positive under partial transposition (PPT entangled states). Our construction also allows to generate numerical (random) examples of PPT entangled states.
	
	Our construction exploits some specific properties of the projectors onto the symmetric and antisymmetric subspace, and in many ways one can consider the states that we put forward as modifications or generalizations of the well-known Werner states. It is worth emphasizing that the construction of Section \ref{sec:analytic} suggests that, when considering examples of noisy entangled states in the context of quantum effects and quantum protocols, it might be interesting to consider a two-parameter family of states
	\begin{equation}
	\label{eq:twoparameterfamily}
	(1-p_\cA-p_\cS)\dyad{\psi_\cA} + p_\cA \frac{P_\cA}{d_\cA} + p_\cS \frac{P_\cS}{d_\cS},
	\end{equation}
	with $p_\cA\geq0$, $p_\cS\geq0$, $p_\cA + p_\cS \leq 1$, and $\ket{\psi_\cA}$ an antisymmetric state vector with Schmidt rank strictly larger than two, so that this two-parameter family comprises all Werner states and also PPT entangled states. Notice that, in even dimensions, if $\ket{\psi_\cA}$ is chosen to be maximally entangled, the family of states \eqref{eq:twoparameterfamily} comprises both Werner states and isotropic states (up to local unitaries).
	
	Future work that takes into account the key properties of the symmetric and antisymmetric subspaces that we have made use of may lead to further generalizations.
	
	It is worth noticing that, thanks to the Choi-Jamio{\l}kowski isomorphism~\cite{Choi:1975aa,Jamiokowski:1972aa}, our construction identifies classes of PPT-binding~\cite{Horodecki:2000aa} but not entanglement-breaking~\cite{Horodecki:2003aa} channels, which might be useful to study superactivation effects in quantum information processing~\cite{Smith1812}.

	\begin{acknowledgments}
	We thank P\'{a}l K\'{a}roly Ferenc for correspondence and for pointing out typos in the preprint version of our manuscript.
	We acknowledge support from European Union's Horizon 2020 Research and Innovation Programme under the Marie Sk{\l}odowska-Curie Action OPERACQC (Grant Agreement No. 661338).
	\end{acknowledgments}


\begin{thebibliography}{33}%
		\makeatletter
		\providecommand \@ifxundefined [1]{%
			\@ifx{#1\undefined}
		}%
		\providecommand \@ifnum [1]{%
			\ifnum #1\expandafter \@firstoftwo
			\else \expandafter \@secondoftwo
			\fi
		}%
		\providecommand \@ifx [1]{%
			\ifx #1\expandafter \@firstoftwo
			\else \expandafter \@secondoftwo
			\fi
		}%
		\providecommand \natexlab [1]{#1}%
		\providecommand \enquote  [1]{``#1''}%
		\providecommand \bibnamefont  [1]{#1}%
		\providecommand \bibfnamefont [1]{#1}%
		\providecommand \citenamefont [1]{#1}%
		\providecommand \href@noop [0]{\@secondoftwo}%
		\providecommand \href [0]{\begingroup \@sanitize@url \@href}%
		\providecommand \@href[1]{\@@startlink{#1}\@@href}%
		\providecommand \@@href[1]{\endgroup#1\@@endlink}%
		\providecommand \@sanitize@url [0]{\catcode `\\12\catcode `\$12\catcode
			`\&12\catcode `\#12\catcode `\^12\catcode `\_12\catcode `\%12\relax}%
		\providecommand \@@startlink[1]{}%
		\providecommand \@@endlink[0]{}%
		\providecommand \url  [0]{\begingroup\@sanitize@url \@url }%
		\providecommand \@url [1]{\endgroup\@href {#1}{\urlprefix }}%
		\providecommand \urlprefix  [0]{URL }%
		\providecommand \Eprint [0]{\href }%
		\providecommand \doibase [0]{http://dx.doi.org/}%
		\providecommand \selectlanguage [0]{\@gobble}%
		\providecommand \bibinfo  [0]{\@secondoftwo}%
		\providecommand \bibfield  [0]{\@secondoftwo}%
		\providecommand \translation [1]{[#1]}%
		\providecommand \BibitemOpen [0]{}%
		\providecommand \bibitemStop [0]{}%
		\providecommand \bibitemNoStop [0]{.\EOS\space}%
		\providecommand \EOS [0]{\spacefactor3000\relax}%
		\providecommand \BibitemShut  [1]{\csname bibitem#1\endcsname}%
		\let\auto@bib@innerbib\@empty
		\bibitem [{\citenamefont {Horodecki}\ \emph {et~al.}(2009)\citenamefont
			{Horodecki}, \citenamefont {Horodecki}, \citenamefont {Horodecki},\ and\
			\citenamefont {Horodecki}}]{Horodecki:2009aa}%
		\BibitemOpen
		\bibfield  {author} {\bibinfo {author} {\bibfnamefont {R.}~\bibnamefont
				{Horodecki}}, \bibinfo {author} {\bibfnamefont {P.}~\bibnamefont
				{Horodecki}}, \bibinfo {author} {\bibfnamefont {M.}~\bibnamefont
				{Horodecki}}, \ and\ \bibinfo {author} {\bibfnamefont {K.}~\bibnamefont
				{Horodecki}},\ }\href {\doibase 10.1103/RevModPhys.81.865} {\bibfield
			{journal} {\bibinfo  {journal} {Rev. Mod. Phys.}\ }\textbf {\bibinfo {volume}
				{81}},\ \bibinfo {pages} {865} (\bibinfo {year} {2009})}\BibitemShut
		{NoStop}%
		\bibitem [{\citenamefont {Nielsen}\ and\ \citenamefont
			{Chuang}(2011)}]{Nielsen:2011aa}%
		\BibitemOpen
		\bibfield  {author} {\bibinfo {author} {\bibfnamefont {M.~A.}\ \bibnamefont
				{Nielsen}}\ and\ \bibinfo {author} {\bibfnamefont {I.~L.}\ \bibnamefont
				{Chuang}},\ }\href@noop {} {\emph {\bibinfo {title} {Quantum Computation and
					Quantum Information: 10th Anniversary Edition}}},\ \bibinfo {edition} {10th}\
		ed.\ (\bibinfo  {publisher} {Cambridge University Press},\ \bibinfo {address}
		{New York, NY, USA},\ \bibinfo {year} {2011})\BibitemShut {NoStop}%
		\bibitem [{\citenamefont {Bennett}\ \emph {et~al.}(1996)\citenamefont
			{Bennett}, \citenamefont {DiVincenzo}, \citenamefont {Smolin},\ and\
			\citenamefont {Wootters}}]{Bennett:1996aa}%
		\BibitemOpen
		\bibfield  {author} {\bibinfo {author} {\bibfnamefont {C.~H.}\ \bibnamefont
				{Bennett}}, \bibinfo {author} {\bibfnamefont {D.~P.}\ \bibnamefont
				{DiVincenzo}}, \bibinfo {author} {\bibfnamefont {J.~A.}\ \bibnamefont
				{Smolin}}, \ and\ \bibinfo {author} {\bibfnamefont {W.~K.}\ \bibnamefont
				{Wootters}},\ }\href {\doibase 10.1103/PhysRevA.54.3824} {\bibfield
			{journal} {\bibinfo  {journal} {Phys. Rev. A}\ }\textbf {\bibinfo {volume}
				{54}},\ \bibinfo {pages} {3824} (\bibinfo {year} {1996})}\BibitemShut
		{NoStop}%
		\bibitem [{\citenamefont {Horodecki}\ \emph {et~al.}(1998)\citenamefont
			{Horodecki}, \citenamefont {Horodecki},\ and\ \citenamefont
			{Horodecki}}]{Horodecki:1998aa}%
		\BibitemOpen
		\bibfield  {author} {\bibinfo {author} {\bibfnamefont {M.}~\bibnamefont
				{Horodecki}}, \bibinfo {author} {\bibfnamefont {P.}~\bibnamefont
				{Horodecki}}, \ and\ \bibinfo {author} {\bibfnamefont {R.}~\bibnamefont
				{Horodecki}},\ }\href {\doibase 10.1103/PhysRevLett.80.5239} {\bibfield
			{journal} {\bibinfo  {journal} {Phys. Rev. Lett.}\ }\textbf {\bibinfo
				{volume} {80}},\ \bibinfo {pages} {5239} (\bibinfo {year}
			{1998})}\BibitemShut {NoStop}%
		\bibitem [{Note1()}]{Note1}%
		\BibitemOpen
		\bibinfo {note} {Whether it is only PPT entangled states that cannot be
			distilled is one of the last major open questions in entanglement
			theory.}\BibitemShut {Stop}%
		\bibitem [{\citenamefont {G{\"u}hne}\ and\ \citenamefont
			{T{\'o}th}(2009)}]{GUHNE20091}%
		\BibitemOpen
		\bibfield  {author} {\bibinfo {author} {\bibfnamefont {O.}~\bibnamefont
				{G{\"u}hne}}\ and\ \bibinfo {author} {\bibfnamefont {G.}~\bibnamefont
				{T{\'o}th}},\ }\href {\doibase
			http://dx.doi.org/10.1016/j.physrep.2009.02.004} {\bibfield  {journal}
			{\bibinfo  {journal} {Physics Reports}\ }\textbf {\bibinfo {volume} {474}},\
			\bibinfo {pages} {1 } (\bibinfo {year} {2009})}\BibitemShut {NoStop}%
		\bibitem [{\citenamefont {Moroder}\ \emph {et~al.}(2014)\citenamefont
			{Moroder}, \citenamefont {Gittsovich}, \citenamefont {Huber},\ and\
			\citenamefont {G\"uhne}}]{Moroder:2014aa}%
		\BibitemOpen
		\bibfield  {author} {\bibinfo {author} {\bibfnamefont {T.}~\bibnamefont
				{Moroder}}, \bibinfo {author} {\bibfnamefont {O.}~\bibnamefont {Gittsovich}},
			\bibinfo {author} {\bibfnamefont {M.}~\bibnamefont {Huber}}, \ and\ \bibinfo
			{author} {\bibfnamefont {O.}~\bibnamefont {G\"uhne}},\ }\href {\doibase
			10.1103/PhysRevLett.113.050404} {\bibfield  {journal} {\bibinfo  {journal}
				{Phys. Rev. Lett.}\ }\textbf {\bibinfo {volume} {113}},\ \bibinfo {pages}
			{050404} (\bibinfo {year} {2014})}\BibitemShut {NoStop}%
		\bibitem [{\citenamefont {V{\'e}rtesi}\ and\ \citenamefont
			{Brunner}(2014)}]{Vertesi:2014aa}%
		\BibitemOpen
		\bibfield  {author} {\bibinfo {author} {\bibfnamefont {T.}~\bibnamefont
				{V{\'e}rtesi}}\ and\ \bibinfo {author} {\bibfnamefont {N.}~\bibnamefont
				{Brunner}},\ }\href {http://dx.doi.org/10.1038/ncomms6297} {\bibfield
			{journal} {\bibinfo  {journal} {Nature Communications}\ }\textbf {\bibinfo
				{volume} {5}},\ \bibinfo {pages} {5297 EP } (\bibinfo {year}
			{2014})}\BibitemShut {NoStop}%
		\bibitem [{\citenamefont {Horodecki}\ \emph {et~al.}(1996)\citenamefont
			{Horodecki}, \citenamefont {Horodecki},\ and\ \citenamefont
			{Horodecki}}]{HORODECKI19961}%
		\BibitemOpen
		\bibfield  {author} {\bibinfo {author} {\bibfnamefont {M.}~\bibnamefont
				{Horodecki}}, \bibinfo {author} {\bibfnamefont {P.}~\bibnamefont
				{Horodecki}}, \ and\ \bibinfo {author} {\bibfnamefont {R.}~\bibnamefont
				{Horodecki}},\ }\href {\doibase
			http://dx.doi.org/10.1016/S0375-9601(96)00706-2} {\bibfield  {journal}
			{\bibinfo  {journal} {Physics Letters A}\ }\textbf {\bibinfo {volume}
				{223}},\ \bibinfo {pages} {1 } (\bibinfo {year} {1996})}\BibitemShut
		{NoStop}%
		\bibitem [{\citenamefont {Bhatia}(2011)}]{bhatia:1997aa}%
		\BibitemOpen
		\bibfield  {author} {\bibinfo {author} {\bibfnamefont {R.}\ \bibnamefont
				{Bhatia}}\ ,\ }\href@noop {} {\emph {\bibinfo {title} {Matrix Analysis}}}.\ (\bibinfo  {publisher} {Springer-Verlag},\ \bibinfo {address}
		{New York, NY, USA},\ \bibinfo {year} {1997})\BibitemShut {NoStop}%
		\bibitem [{\citenamefont {Smith}\ and\ \citenamefont {Yard}(2008)}]{Smith1812}%
		\BibitemOpen
		\bibfield  {author} {\bibinfo {author} {\bibfnamefont {G.}~\bibnamefont
				{Smith}}\ and\ \bibinfo {author} {\bibfnamefont {J.}~\bibnamefont {Yard}},\
		}\href {\doibase 10.1126/science.1162242} {\bibfield  {journal} {\bibinfo
				{journal} {Science}\ }\textbf {\bibinfo {volume} {321}},\ \bibinfo {pages}
			{1812} (\bibinfo {year} {2008})}\BibitemShut {NoStop}%
		\bibitem [{\citenamefont {Horodecki}(1997)}]{HORODECKI1997333}%
		\BibitemOpen
		\bibfield  {author} {\bibinfo {author} {\bibfnamefont {P.}~\bibnamefont
				{Horodecki}},\ }\href {\doibase
			http://dx.doi.org/10.1016/S0375-9601(97)00416-7} {\bibfield  {journal}
			{\bibinfo  {journal} {Physics Letters A}\ }\textbf {\bibinfo {volume}
				{232}},\ \bibinfo {pages} {333 } (\bibinfo {year} {1997})}\BibitemShut
		{NoStop}%
		\bibitem [{\citenamefont {Bennett}\ \emph {et~al.}(1999)\citenamefont
			{Bennett}, \citenamefont {DiVincenzo}, \citenamefont {Mor}, \citenamefont
			{Shor}, \citenamefont {Smolin},\ and\ \citenamefont
			{Terhal}}]{Bennett:1999aa}%
		\BibitemOpen
		\bibfield  {author} {\bibinfo {author} {\bibfnamefont {C.~H.}\ \bibnamefont
				{Bennett}}, \bibinfo {author} {\bibfnamefont {D.~P.}\ \bibnamefont
				{DiVincenzo}}, \bibinfo {author} {\bibfnamefont {T.}~\bibnamefont {Mor}},
			\bibinfo {author} {\bibfnamefont {P.~W.}\ \bibnamefont {Shor}}, \bibinfo
			{author} {\bibfnamefont {J.~A.}\ \bibnamefont {Smolin}}, \ and\ \bibinfo
			{author} {\bibfnamefont {B.~M.}\ \bibnamefont {Terhal}},\ }\href {\doibase
			10.1103/PhysRevLett.82.5385} {\bibfield  {journal} {\bibinfo  {journal}
				{Phys. Rev. Lett.}\ }\textbf {\bibinfo {volume} {82}},\ \bibinfo {pages}
			{5385} (\bibinfo {year} {1999})}\BibitemShut {NoStop}%
		\bibitem [{\citenamefont {Benatti}\ \emph {et~al.}(2004)\citenamefont
			{Benatti}, \citenamefont {Floreanini},\ and\ \citenamefont
			{Piani}}]{Benatti:2004aa}%
		\BibitemOpen
		\bibfield  {author} {\bibinfo {author} {\bibfnamefont {F.}~\bibnamefont
				{Benatti}}, \bibinfo {author} {\bibfnamefont {R.}~\bibnamefont {Floreanini}},
			\ and\ \bibinfo {author} {\bibfnamefont {M.}~\bibnamefont {Piani}},\ }\href
		{\doibase 10.1007/s11080-004-6622-6} {\bibfield  {journal} {\bibinfo
				{journal} {Open Systems {\&} Information Dynamics}\ }\textbf {\bibinfo
				{volume} {11}},\ \bibinfo {pages} {325} (\bibinfo {year} {2004})}\BibitemShut
		{NoStop}%
		\bibitem [{\citenamefont {Piani}\ and\ \citenamefont
			{Mora}(2007)}]{Piani:2007aa}%
		\BibitemOpen
		\bibfield  {author} {\bibinfo {author} {\bibfnamefont {M.}~\bibnamefont
				{Piani}}\ and\ \bibinfo {author} {\bibfnamefont {C.~E.}\ \bibnamefont
				{Mora}},\ }\href {\doibase 10.1103/PhysRevA.75.012305} {\bibfield  {journal}
			{\bibinfo  {journal} {Phys. Rev. A}\ }\textbf {\bibinfo {volume} {75}},\
			\bibinfo {pages} {012305} (\bibinfo {year} {2007})}\BibitemShut {NoStop}%
			\bibitem [{\citenamefont {Sentis}\ \emph {et~al.}(2016)\citenamefont
			{Sent\'{\i}s}, \citenamefont {Eltschka},\ and\ \citenamefont
			{Siewert}}]{Sentis:2016aa}%
		\BibitemOpen
		\bibfield  {author} {\bibinfo {author} {\bibfnamefont {G.}~\bibnamefont
				{Sent\'{\i}s}}, \bibinfo {author} {\bibfnamefont {C.}~\bibnamefont {Eltschka}},
			\ and\ \bibinfo {author} {\bibfnamefont {J.}~\bibnamefont {Siewert}},\ }\href
		{\doibase 10.1103/PhysRevA.94.020302} {\bibfield  {journal} {\bibinfo
				{journal} {Phys. Rev. A}\ }\textbf {\bibinfo
				{volume} {94}},\ \bibinfo {pages} {020302} (\bibinfo {year} {2016})}\BibitemShut
		{NoStop}%
		\bibitem [{\citenamefont {Werner}(1989)}]{Werner:1989aa}%
		\BibitemOpen
		\bibfield  {author} {\bibinfo {author} {\bibfnamefont {R.~F.}\ \bibnamefont
				{Werner}},\ }\href {\doibase 10.1103/PhysRevA.40.4277} {\bibfield  {journal}
			{\bibinfo  {journal} {Phys. Rev. A}\ }\textbf {\bibinfo {volume} {40}},\
			\bibinfo {pages} {4277} (\bibinfo {year} {1989})}\BibitemShut {NoStop}%
		\bibitem [{\citenamefont {Horodecki}\ and\ \citenamefont
			{Horodecki}(1999)}]{Horodecki:1999aa}%
		\BibitemOpen
		\bibfield  {author} {\bibinfo {author} {\bibfnamefont {M.}~\bibnamefont
				{Horodecki}}\ and\ \bibinfo {author} {\bibfnamefont {P.}~\bibnamefont
				{Horodecki}},\ }\href {\doibase 10.1103/PhysRevA.59.4206} {\bibfield
			{journal} {\bibinfo  {journal} {Phys. Rev. A}\ }\textbf {\bibinfo {volume}
				{59}},\ \bibinfo {pages} {4206} (\bibinfo {year} {1999})}\BibitemShut
		{NoStop}%
		\bibitem [{\citenamefont {Amselem}\ and\ \citenamefont
			{Bourennane}(2009)}]{Amselem:2009aa}%
		\BibitemOpen
		\bibfield  {author} {\bibinfo {author} {\bibfnamefont {E.}~\bibnamefont
				{Amselem}}\ and\ \bibinfo {author} {\bibfnamefont {M.}~\bibnamefont
				{Bourennane}},\ }\href {https://search.proquest.com/docview/194661987?
			accountid=14116} {\bibfield  {journal} {\bibinfo  {journal} {Nature Physics}\
			}\textbf {\bibinfo {volume} {5}},\ \bibinfo {pages} {748} (\bibinfo {year}
			{2009})},\ \bibinfo {note} {copyright - Copyright Nature Publishing Group Oct
			2009; Last updated - 2012-11-20}\BibitemShut {NoStop}%
		\bibitem [{\citenamefont {Lavoie}\ \emph {et~al.}(2010)\citenamefont {Lavoie},
			\citenamefont {Kaltenbaek}, \citenamefont {Piani},\ and\ \citenamefont
			{Resch}}]{Lavoie:2010aa}%
		\BibitemOpen
		\bibfield  {author} {\bibinfo {author} {\bibfnamefont {J.}~\bibnamefont
				{Lavoie}}, \bibinfo {author} {\bibfnamefont {R.}~\bibnamefont {Kaltenbaek}},
			\bibinfo {author} {\bibfnamefont {M.}~\bibnamefont {Piani}}, \ and\ \bibinfo
			{author} {\bibfnamefont {K.~J.}\ \bibnamefont {Resch}},\ }\href {\doibase
			10.1103/PhysRevLett.105.130501} {\bibfield  {journal} {\bibinfo  {journal}
				{Phys. Rev. Lett.}\ }\textbf {\bibinfo {volume} {105}},\ \bibinfo {pages}
			{130501} (\bibinfo {year} {2010})}\BibitemShut {NoStop}%
		\bibitem [{\citenamefont {Barreiro}\ \emph {et~al.}(2010)\citenamefont
			{Barreiro}, \citenamefont {Schindler}, \citenamefont {Guhne}, \citenamefont
			{Monz}, \citenamefont {Chwalla}, \citenamefont {Roos}, \citenamefont
			{Hennrich},\ and\ \citenamefont {Blatt}}]{Barreiro:2010aa}%
		\BibitemOpen
		\bibfield  {author} {\bibinfo {author} {\bibfnamefont {J.~T.}\ \bibnamefont
				{Barreiro}}, \bibinfo {author} {\bibfnamefont {P.}~\bibnamefont {Schindler}},
			\bibinfo {author} {\bibfnamefont {O.}~\bibnamefont {Guhne}}, \bibinfo
			{author} {\bibfnamefont {T.}~\bibnamefont {Monz}}, \bibinfo {author}
			{\bibfnamefont {M.}~\bibnamefont {Chwalla}}, \bibinfo {author} {\bibfnamefont
				{C.~F.}\ \bibnamefont {Roos}}, \bibinfo {author} {\bibfnamefont
				{M.}~\bibnamefont {Hennrich}}, \ and\ \bibinfo {author} {\bibfnamefont
				{R.}~\bibnamefont {Blatt}},\ }\href {http://dx.doi.org/10.1038/nphys1781}
		{\bibfield  {journal} {\bibinfo  {journal} {Nat Phys}\ }\textbf {\bibinfo
				{volume} {6}},\ \bibinfo {pages} {943} (\bibinfo {year} {2010})}\BibitemShut
		{NoStop}%
		\bibitem [{\citenamefont {DiGuglielmo}\ \emph {et~al.}(2011)\citenamefont
			{DiGuglielmo}, \citenamefont {Samblowski}, \citenamefont {Hage},
			\citenamefont {Pineda}, \citenamefont {Eisert},\ and\ \citenamefont
			{Schnabel}}]{DiGuglielmo:2011aa}%
		\BibitemOpen
		\bibfield  {author} {\bibinfo {author} {\bibfnamefont {J.}~\bibnamefont
				{DiGuglielmo}}, \bibinfo {author} {\bibfnamefont {A.}~\bibnamefont
				{Samblowski}}, \bibinfo {author} {\bibfnamefont {B.}~\bibnamefont {Hage}},
			\bibinfo {author} {\bibfnamefont {C.}~\bibnamefont {Pineda}}, \bibinfo
			{author} {\bibfnamefont {J.}~\bibnamefont {Eisert}}, \ and\ \bibinfo {author}
			{\bibfnamefont {R.}~\bibnamefont {Schnabel}},\ }\href {\doibase
			10.1103/PhysRevLett.107.240503} {\bibfield  {journal} {\bibinfo  {journal}
				{Phys. Rev. Lett.}\ }\textbf {\bibinfo {volume} {107}},\ \bibinfo {pages}
			{240503} (\bibinfo {year} {2011})}\BibitemShut {NoStop}%
		\bibitem [{\citenamefont {Terhal}\ and\ \citenamefont
			{Horodecki}(2000)}]{Terhal:2000aa}%
		\BibitemOpen
		\bibfield  {author} {\bibinfo {author} {\bibfnamefont {B.~M.}\ \bibnamefont
				{Terhal}}\ and\ \bibinfo {author} {\bibfnamefont {P.}~\bibnamefont
				{Horodecki}},\ }\href {\doibase 10.1103/PhysRevA.61.040301} {\bibfield
			{journal} {\bibinfo  {journal} {Phys. Rev. A}\ }\textbf {\bibinfo {volume}
				{61}},\ \bibinfo {pages} {040301} (\bibinfo {year} {2000})}\BibitemShut
		{NoStop}%
		\bibitem [{\citenamefont {Peres}(1996)}]{Peres:1996aa}%
		\BibitemOpen
		\bibfield  {author} {\bibinfo {author} {\bibfnamefont {A.}~\bibnamefont
				{Peres}},\ }\href {\doibase 10.1103/PhysRevLett.77.1413} {\bibfield
			{journal} {\bibinfo  {journal} {Phys. Rev. Lett.}\ }\textbf {\bibinfo
				{volume} {77}},\ \bibinfo {pages} {1413} (\bibinfo {year}
			{1996})}\BibitemShut {NoStop}%
		\bibitem [{\citenamefont {DiVincenzo}\ \emph {et~al.}(2000)\citenamefont
			{DiVincenzo}, \citenamefont {Shor}, \citenamefont {Smolin}, \citenamefont
			{Terhal},\ and\ \citenamefont {Thapliyal}}]{DiVincenzo:2000aa}%
		\BibitemOpen
		\bibfield  {author} {\bibinfo {author} {\bibfnamefont {D.~P.}\ \bibnamefont
				{DiVincenzo}}, \bibinfo {author} {\bibfnamefont {P.~W.}\ \bibnamefont
				{Shor}}, \bibinfo {author} {\bibfnamefont {J.~A.}\ \bibnamefont {Smolin}},
			\bibinfo {author} {\bibfnamefont {B.~M.}\ \bibnamefont {Terhal}}, \ and\
			\bibinfo {author} {\bibfnamefont {A.~V.}\ \bibnamefont {Thapliyal}},\ }\href
		{\doibase 10.1103/PhysRevA.61.062312} {\bibfield  {journal} {\bibinfo
				{journal} {Phys. Rev. A}\ }\textbf {\bibinfo {volume} {61}},\ \bibinfo
			{pages} {062312} (\bibinfo {year} {2000})}\BibitemShut {NoStop}%
		\bibitem [{\citenamefont {Shor}\ \emph {et~al.}(2001)\citenamefont {Shor},
			\citenamefont {Smolin},\ and\ \citenamefont {Terhal}}]{Shor:2001aa}%
		\BibitemOpen
		\bibfield  {author} {\bibinfo {author} {\bibfnamefont {P.~W.}\ \bibnamefont
				{Shor}}, \bibinfo {author} {\bibfnamefont {J.~A.}\ \bibnamefont {Smolin}}, \
			and\ \bibinfo {author} {\bibfnamefont {B.~M.}\ \bibnamefont {Terhal}},\
		}\href {\doibase 10.1103/PhysRevLett.86.2681} {\bibfield  {journal} {\bibinfo
				{journal} {Phys. Rev. Lett.}\ }\textbf {\bibinfo {volume} {86}},\ \bibinfo
			{pages} {2681} (\bibinfo {year} {2001})}\BibitemShut {NoStop}%
		\bibitem [{\citenamefont {Pankowski}\ \emph {et~al.}(2010)\citenamefont
			{Pankowski}, \citenamefont {Piani}, \citenamefont {Horodecki},\ and\
			\citenamefont {Horodecki}}]{Pankowski:2010aa}%
		\BibitemOpen
		\bibfield  {author} {\bibinfo {author} {\bibfnamefont {{\L}.}~\bibnamefont
				{Pankowski}}, \bibinfo {author} {\bibfnamefont {M.}~\bibnamefont {Piani}},
			\bibinfo {author} {\bibfnamefont {M.}~\bibnamefont {Horodecki}}, \ and\
			\bibinfo {author} {\bibfnamefont {P.}~\bibnamefont {Horodecki}},\ }\href
		{\doibase 10.1109/TIT.2010.2050810} {\bibfield  {journal} {\bibinfo
				{journal} {IEEE Transactions on Information Theory}\ }\textbf {\bibinfo
				{volume} {56}},\ \bibinfo {pages} {4085} (\bibinfo {year}
			{2010})}\BibitemShut {NoStop}%
		\bibitem [{\citenamefont {Lancien}\ \emph {et~al.}(2016)\citenamefont
			{Lancien}, \citenamefont {Di~Martino}, \citenamefont {Huber}, \citenamefont
			{Piani}, \citenamefont {Adesso},\ and\ \citenamefont
			{Winter}}]{Lancien:2016aa}%
		\BibitemOpen
		\bibfield  {author} {\bibinfo {author} {\bibfnamefont {C.}~\bibnamefont
				{Lancien}}, \bibinfo {author} {\bibfnamefont {S.}~\bibnamefont {Di~Martino}},
			\bibinfo {author} {\bibfnamefont {M.}~\bibnamefont {Huber}}, \bibinfo
			{author} {\bibfnamefont {M.}~\bibnamefont {Piani}}, \bibinfo {author}
			{\bibfnamefont {G.}~\bibnamefont {Adesso}}, \ and\ \bibinfo {author}
			{\bibfnamefont {A.}~\bibnamefont {Winter}},\ }\href {\doibase
			10.1103/PhysRevLett.117.060501} {\bibfield  {journal} {\bibinfo  {journal}
				{Phys. Rev. Lett.}\ }\textbf {\bibinfo {volume} {117}},\ \bibinfo {pages}
			{060501} (\bibinfo {year} {2016})}\BibitemShut {NoStop}%
		\bibitem [{\citenamefont {DiVincenzo}\ \emph {et~al.}(2002)\citenamefont
			{DiVincenzo}, \citenamefont {Leung},\ and\ \citenamefont
			{Terhal}}]{DiVincenzo:2002aa}%
		\BibitemOpen
		\bibfield  {author} {\bibinfo {author} {\bibfnamefont {D.~P.}\ \bibnamefont
				{DiVincenzo}}, \bibinfo {author} {\bibfnamefont {D.~W.}\ \bibnamefont
				{Leung}}, \ and\ \bibinfo {author} {\bibfnamefont {B.~M.}\ \bibnamefont
				{Terhal}},\ }\href {\doibase 10.1109/18.985948} {\bibfield  {journal}
			{\bibinfo  {journal} {IEEE Transactions on Information Theory}\ }\textbf
			{\bibinfo {volume} {48}},\ \bibinfo {pages} {580} (\bibinfo {year}
			{2002})}\BibitemShut {NoStop}%
		\bibitem [{\citenamefont {DiVincenzo}\ \emph {et~al.}(2003)\citenamefont
			{DiVincenzo}, \citenamefont {Hayden},\ and\ \citenamefont
			{Terhal}}]{DiVincenzo2003}%
		\BibitemOpen
		\bibfield  {author} {\bibinfo {author} {\bibfnamefont {D.~P.}\ \bibnamefont
				{DiVincenzo}}, \bibinfo {author} {\bibfnamefont {P.}~\bibnamefont {Hayden}},
			\ and\ \bibinfo {author} {\bibfnamefont {B.~M.}\ \bibnamefont {Terhal}},\
		}\href {\doibase 10.1023/A:1026013201376} {\bibfield  {journal} {\bibinfo
				{journal} {Foundations of Physics}\ }\textbf {\bibinfo {volume} {33}},\
			\bibinfo {pages} {1629} (\bibinfo {year} {2003})}\BibitemShut {NoStop}%
		\bibitem [{Note2()}]{Note2}%
		\BibitemOpen
		\bibinfo {note} {Notice that all antisymmetric mixed states have at least
			Schmidt number two, because the antisymmetric subspace does not contain
			product states, as it can be verified by using (\ref
			{eq:propPaalphabeta}).}\BibitemShut {Stop}%
		\bibitem [{\citenamefont {Choi}(1975)}]{Choi:1975aa}%
		\BibitemOpen
		\bibfield  {author} {\bibinfo {author} {\bibfnamefont {M.-D.}\ \bibnamefont
				{Choi}},\ }\href {\doibase http://dx.doi.org/10.1016/0024-3795(75)90075-0}
		{\bibfield  {journal} {\bibinfo  {journal} {Linear Algebra and its
					Applications}\ }\textbf {\bibinfo {volume} {10}},\ \bibinfo {pages} {285 }
			(\bibinfo {year} {1975})}\BibitemShut {NoStop}%
		\bibitem [{\citenamefont {Jamio{\l}kowski}(1972)}]{Jamiokowski:1972aa}%
		\BibitemOpen
		\bibfield  {author} {\bibinfo {author} {\bibfnamefont {A.}~\bibnamefont
				{Jamio{\l}kowski}},\ }\href {\doibase
			http://dx.doi.org/10.1016/0034-4877(72)90011-0} {\bibfield  {journal}
			{\bibinfo  {journal} {Reports on Mathematical Physics}\ }\textbf {\bibinfo
				{volume} {3}},\ \bibinfo {pages} {275 } (\bibinfo {year} {1972})}\BibitemShut
		{NoStop}%
		\bibitem [{\citenamefont {Horodecki}\ \emph {et~al.}(2000)\citenamefont
			{Horodecki}, \citenamefont {Horodecki},\ and\ \citenamefont
			{Horodecki}}]{Horodecki:2000aa}%
		\BibitemOpen
		\bibfield  {author} {\bibinfo {author} {\bibfnamefont {P.}~\bibnamefont
				{Horodecki}}, \bibinfo {author} {\bibfnamefont {M.}~\bibnamefont
				{Horodecki}}, \ and\ \bibinfo {author} {\bibfnamefont {R.}~\bibnamefont
				{Horodecki}},\ }\href {\doibase 10.1080/09500340008244047} {\bibfield
			{journal} {\bibinfo  {journal} {Journal of Modern Optics}\ }\textbf {\bibinfo
				{volume} {47}},\ \bibinfo {pages} {347} (\bibinfo {year} {2000})}\BibitemShut
		{NoStop}%
		\bibitem [{\citenamefont {Horodecki}\ \emph {et~al.}(2003)\citenamefont
			{Horodecki}, \citenamefont {Shor},\ and\ \citenamefont
			{Ruskai}}]{Horodecki:2003aa}%
		\BibitemOpen
		\bibfield  {author} {\bibinfo {author} {\bibfnamefont {M.}~\bibnamefont
				{Horodecki}}, \bibinfo {author} {\bibfnamefont {P.~W.}\ \bibnamefont {Shor}},
			\ and\ \bibinfo {author} {\bibfnamefont {M.~B.}\ \bibnamefont {Ruskai}},\
		}\href {\doibase 10.1142/S0129055X03001709} {\bibfield  {journal} {\bibinfo
				{journal} {Reviews in Mathematical Physics}\ }\textbf {\bibinfo {volume}
				{15}},\ \bibinfo {pages} {629} (\bibinfo {year} {2003})}\BibitemShut
		{NoStop}%
	\end{thebibliography}
\end{document}